\documentclass[10pt,letterpaper]{article}
\usepackage{spconf,amsmath,amssymb,amsthm,booktabs,array,enumitem,microtype}
\usepackage{cite}
\usepackage[T1]{fontenc}
\usepackage{xcolor}
\IfFileExists{algorithm.sty}{\usepackage{algorithm}}{\usepackage{float}\floatstyle{ruled}\newfloat{algorithm}{tbp}{loa}\floatname{algorithm}{Algorithm}}
\usepackage[hidelinks]{hyperref}
\newtheoremstyle{paperplain}{4pt}{4pt}{\itshape}{} {\bfseries}{.}{.5em}{}
\newtheoremstyle{paperdefinition}{4pt}{4pt}{}{} {\bfseries}{.}{.5em}{}
\theoremstyle{paperplain}
\newtheorem{theorem}{Theorem}
\newtheorem{prop}{Proposition}
\theoremstyle{paperdefinition}
\newtheorem{example}{Example}
\newtheorem{definition}{Definition}
\newcommand{\F}{\mathbb F_q}
\newcommand{\E}{\mathbb E}
\newcommand{\Var}{\operatorname{Var}}
\newcommand{\rank}{\operatorname{rank}}
\newcommand{\Ball}{\operatorname{Ball}}
\newcommand{\TV}{d_{\mathrm{TV}}}
\newcommand{\rh}{\rho_{\mathrm m}}
\title{Low-Rank Masking for Single-Server Matrix Multiplication}
\name{Alejandro Cohen$^{1,2}$, Rafael G. L. D'Oliveira$^3$, and Alex Sprintson$^4$}
\address{$^1$Faculty of Electrical and Computer Engineering, Technion, Israel\\ $^2$Department of Electronic Systems, Aalborg University, Copenhagen Campus, Denmark\\$^3$School of Mathematical and Statistical Sciences, Clemson University, USA\\$^4$Department of Electrical and Computer Engineering, George Mason University, USA}
\hypersetup{pdftitle={Low-Rank Masking for Single-Server Matrix Multiplication},pdfauthor={Alejandro Cohen; Rafael G. L. D'Oliveira; Alex Sprintson}}
\begin{document}
\maketitle

\begin{abstract}
We study the statistical privacy of outsourcing matrix multiplication over a finite field $\F$ to a single server using additive masks of rank at most $r$. For independent uniform $n\times n$ inputs, we show that uniform \emph{rank-ball masks} and products of independent uniform factors give maximal-correlation secrecy of at most $q^{-r}$ against the complete server view, with $O(n^2r)$ field operations for encoding and decoding. This secrecy captures how effectively the server is prevented from estimating functions of the inputs. We prove an asymptotically matching lower bound of this secrecy measure for $r=o(n)$, showing that both sampling methods are asymptotically optimal among input-independent additive masks of rank at most $r$, even when secret invertible transformations are allowed. We also characterize the posterior distribution for uniform rank-ball masks under arbitrary joint input distributions and prove approximate individual security for rows and columns under independent uniform inputs. Finally, we show that every input-independent additive mask of rank at most $r=o(n)$ requires $\delta\to1$ in entry-level $(\varepsilon,\delta)$-differential privacy for fixed field size $q$ and bounded $\varepsilon$.
\end{abstract}
\begin{keywords}
Matrix multiplication, secure outsourcing, distributed computing, individual security.
\end{keywords}

\vspace{-0.2cm}
\section{Introduction}
\label{sec:intro}
\vspace{-0.1cm}
We consider a user who has two matrices $A\in\F^{m\times n}$ and $B\in\F^{n\times p}$ and wishes to compute $A\cdot B$ with the assistance of a single server. The challenge is to limit what the server learns about $A$ and $B$ while keeping the user's encoding and recovery costs substantially below the cost of computing $AB$ locally.

Homomorphic encryption and specialized cryptographic protocols can protect inputs when outsourcing to a single server, but their privacy guarantees rely on computational hardness assumptions~\cite{bae2026fast,10.1007/978-3-032-12287-2_4}. Secure distributed matrix multiplication (SDMM) achieves perfect information-theoretic privacy against computationally unbounded servers, but requires multiple servers and a bound on how many may collude~\cite{d2026function}. The privacy guarantees of these SDMM schemes do not hold when a single server receives all the encoded inputs.

This motivates the transformation-based outsourcing approach~\cite{atallah1996secure,ATALLAH2002215,10.1145/1755688.1755695}, in which the user masks the inputs before sending them to the server and recovers the desired product from the returned result. For example, the user can upload $P_1AP_2^{-1}$ and $P_2BP_3^{-1}$, where $P_1,P_2,P_3$ are secret invertible matrices. The server then returns $P_1ABP_3^{-1}$. Using permutations and scalings makes encoding and recovery inexpensive. This line of work includes sparse-matrix protocols~\cite{8458047}, invertible transformations~\cite{LEI2014205,8700179,liu2021efficient,10318946}, and additive masks~\cite{atallah1996secure,10.1145/3055259.3055263,kumar2018secure}.

We focus on low-rank additive masking, which has been used for linear systems~\cite{7218476}, matrix multiplication~\cite{10.1145/3055259.3055263,9300174}, and secure inference~\cite{xiong2025loro}. Blockwise variants use outer-product masks~\cite{kumar2018secure}. For $n\times n$ inputs, each mask of rank at most $r$ can be written as the product of an $n\times r$ matrix and an $r\times n$ matrix. Using these factors, the user can encode the inputs and recover their product with $O(n^2r)$ field operations, while the server performs one matrix multiplication.

We then ask what statistical privacy these masks can provide. We first consider differential privacy (DP)~\cite{dwork2014algorithmic}, which limits how much the distribution of the server's observation can change when one input entry changes. However, Proposition~\ref{prop:dp} shows that every input-independent mask of rank at most $r=o(n)$ requires $\delta\to1$ in entry-level $(\varepsilon,\delta)$-differential privacy when $q$ is fixed and $\varepsilon$ is bounded.

Thus, we study partial privacy guarantees. We consider two sampling methods: multiplying independent uniform factors and choosing a mask uniformly from all matrices of rank at most $r$. For each, we describe which inputs remain possible after the server sees the uploads and how likely they are.

Each input must lie within rank distance $r$ of its upload (Definition~\ref{def:rankball}). For uniform rank-ball masks, the uploads rule out incompatible input pairs while preserving the relative probabilities of those that remain. Theorem~\ref{thm:posterior} establishes this for arbitrary joint input distributions. For independent uniform matrices, Theorem~\ref{thm:rows} gives an approximate form of individual security~\cite{8466890} for rows and columns.

Our main result concerns maximal-correlation secrecy~\cite{8316898}, which quantifies how well the server can estimate functions of the inputs. For independent uniform square inputs, Theorem~\ref{thm:rho} gives maximal correlation at most $q^{-r}$ against the complete server view for both sampling methods, with equality for independent uniform factors. Theorem~\ref{thm:converse} gives an asymptotically matching lower bound for every input-independent additive mask of rank at most $r$, even when combined with secret invertible transformations. Thus, for fixed $q$ and $r=o(n)$, both sampling methods achieve asymptotically optimal maximal-correlation secrecy within this class. 

To our knowledge, this is the first work to establish asymptotically matching achievability and converse bounds for maximal-correlation secrecy of input-independent additive matrix masks under a rank constraint with uniform inputs.

\vspace{-0.2cm}
\subsection{Related Work}

Other single-server work considers adjustable security and efficiency~\cite{zhao2025performance}. The scheme in~\cite{chiang2026mosaic} combines low-rank factors with additional noise to obtain computational security and allows approximate recovery. The additional noise means that the complete masks need not satisfy the rank constraint studied here. Over the real and complex numbers, analog SDMM studies numerical stability and mutual-information leakage~\cite{11505934}, while differentially private distributed multiplication studies privacy--accuracy tradeoffs~\cite{9834493,10206949}. In our setting, the user recovers the finite-field product exactly, and we study the uncertainty that remains about the inputs. Individual secrecy protects each message separately while allowing information about relationships between messages. This notion has been studied in secure network coding~\cite{8466890}, distributed storage~\cite{kadhe2014weakly}, and multi-secret sharing~\cite{10619205}.

\vspace{-0.2cm}
\section{Main Results}
\label{sec:results}
We study privacy against a computationally unbounded honest-but-curious (semi-honest \cite{8700179}) server. Correctness assumes that the server follows the prescribed computation; verifiability against malicious computation is outside the scope of this work. Proofs appear in Section~\ref{sec:proofs}.

\vspace{-0.2cm}
\subsection{Low-Rank Masking}
We consider a user who wishes to multiply $A\in\F^{m\times n}$ and $B\in\F^{n\times p}$ with the help of a single server. The user adds a low-rank mask to each matrix before sending it to the server. Let $1\le r\le\min\{m,n,p\}$ and $R=UV$ and $S=QW$, where $U\in\F^{m\times r}$, $V\in\F^{r\times n}$, $Q\in\F^{n\times r}$, and $W\in\F^{r\times p}$. The masks are kept secret and sampled independently of the inputs and of each other, with fresh masks for each multiplication. Algorithm~\ref{alg:outsourcing} describes the protocol.

\begin{algorithm}[ht]
\caption{Single-server multiplication with low-rank masks}
\label{alg:outsourcing}
\textbf{Input:} $A,B$ and rank parameter $r$.\\
\textbf{Output:} $AB$.
\begin{enumerate}[leftmargin=*,itemsep=1pt,topsep=3pt]
\item \textbf{Precompute.} Sample masks $R=UV$ and $S=QW$. Store the factors, $R,S$, and $P=RS$.
\item \textbf{Upload.} Send $X=A+R$ and $Y=B+S$.
\item \textbf{Compute.} The server returns $Z=XY$.
\item \textbf{Decode.} Return $Z-(AQ)W-U(VB)-P$.
\end{enumerate}
\end{algorithm}

\begin{theorem}\label{thm:complexity}
Given the mask factors, Algorithm~\ref{alg:outsourcing} computes $AB$ using $O(mnr+npr+mpr)$ field operations at the user in the worst case and one matrix multiplication at the server.
\end{theorem}

\begin{example}\label{ex:square}
For $n\times n$ inputs, decoding computes $AQ$, $(AQ)W$, $VB$, and $U(VB)$, each using $n^2r$ field multiplications. Decoding therefore requires $4n^2r$ multiplications, while encoding requires $O(n^2)$ additions. Precomputing $R$ and $S$ uses $2n^2r$ multiplications, and computing $RS$ as $U((VQ)W)$ uses another $n^2r+2nr^2$. Thus, the user's computation is quadratic in $n$ for constant $r$ and $o(n^3)$ when $r=o(n)$.
\end{example}

\subsection{Privacy within Rank Balls}
\label{sec:regions}

We now ask which inputs remain possible after the server sees an upload and how likely they are. Since $X=A+R$ and $\rank(R)\le r$, observing $X=x$ tells the server that $\rank(x-A)\le r$. We use the rank distance.

\begin{definition}\label{def:rankball}
The rank distance between $a,x\in\F^{d\times e}$ is $\rank(x-a)$. The rank ball of radius $r$ centered at $x$ is $\Ball_r(x)=\{a\in\F^{d\times e}:\rank(x-a)\le r\}$.
\end{definition}

We consider two ways to sample the masks. \emph{Low-Rank Ball} samples $R$ uniformly from $\Ball_r(0)$ and keeps a factorization $R=UV$ for decoding. \emph{Low-Rank Factors} samples all entries of $U,V$ independently and uniformly and sets $R=UV$. Both methods can produce every mask of rank at most $r$, but assign different probabilities to them.

After observing $X=x$, the conditional probability of each input $a$ is proportional to $\Pr[A=a]\Pr[R=x-a]$. For uniform rank-ball masks, the second factor is the same for every $a\in\Ball_r(x)$, giving the following result.

\begin{theorem}\label{thm:posterior}
Let $A$ have any distribution and let $R$ be a Low-Rank Ball mask independent of $A$. For every upload $x$ with positive probability,
\begin{equation}\label{eq:posterior}
\Pr[A=a\mid X=x]=\frac{\Pr[A=a]\,\mathbf1\{a\in\Ball_r(x)\}}{\Pr[A\in\Ball_r(x)]}.
\end{equation}
For $A,B$ with any joint distribution and Low-Rank Ball masks sampled independently of each other and of $(A,B)$, the posterior after uploads $x,y$ of positive joint probability is the joint prior restricted to $\Ball_r(x)\times\Ball_r(y)$ and normalized.
\end{theorem}

For uniform rank-ball masks, an upload $x$ rules out inputs at rank distance greater than $r$ from $x$ and preserves the relative probabilities of those that remain. If only one of these inputs has positive prior probability, the server can identify it. Two inputs within rank distance $r$ of the same upload can still have different upload distributions, since each produces uploads in a rank ball centered at itself.

\vspace{-0.2cm}
\subsection{Limits on Differential Privacy}
\label{sec:dp}
We now show that low-rank masking cannot provide useful differential privacy when $r=o(n)$. This limitation comes from the rank bound and holds for every input-independent mask distribution.

\begin{definition}[\cite{dwork2014algorithmic} ]\label{def:dp}
Two inputs are neighbors if they differ in exactly one entry. For $\varepsilon\ge0$ and $\delta\in[0,1]$, the upload $X$ satisfies entry-level $(\varepsilon,\delta)$-differential privacy if
\[
\Pr[X\in\mathcal E\mid A=a]\le e^\varepsilon\Pr[X\in\mathcal E\mid A=a']+\delta
\]
for every set of uploads $\mathcal E$ and every pair of neighboring inputs $a,a'$. The case $\delta=0$ is called pure differential privacy.
\end{definition}

The following bound holds for every input-independent mask distribution with rank at most $r$.

\begin{prop}\label{prop:dp}
Let $K\in\F^{n\times n}$ be independent of the input, with $\rank(K)\le r<n$ almost surely. If $A+K$ satisfies entry-level $(\varepsilon,\delta)$-differential privacy, then
\begin{equation}\label{eq:dp}
\delta\ge1-\frac{(e^\varepsilon+q-1)\log_2q}{q-1}\frac rn.
\end{equation}
No such mask gives pure differential privacy for any finite $\varepsilon$.
\end{prop}

For fixed $q$ and bounded $\varepsilon$, the bound forces $\delta\to1$ when $r=o(n)$. At $\delta=1$, the privacy condition places no restriction on the upload distribution. Thus, no choice of input-independent mask distribution can give useful differential privacy in this regime.

\vspace{-0.2cm}
\subsection{Approximate Individual Security}
\label{sec:individual}
We next ask how much the upload reveals about each row or column separately. We take $A$ to be uniform over $\F^{n\times n}$, let $1\le r<n$, and write $X=A+R$, where $R$ is independent of $A$. We regard the rows $M_i=A_{i,:}\in\F^n$ as independent uniform messages.

\begin{definition}[\cite{8466890}]\label{def:individual}
For $T\subseteq[n]$, let $M_T=(M_i)_{i\in T}$. The messages $M_1,\ldots,M_n$ satisfy $T$-individual security with respect to $X$ if $M_T$ is independent of $X$, that is, $I(M_T;X)=0$.
\end{definition}

This protects the rows indexed by $T$ jointly, including relationships among them. We relax independence using total variation, which measures the largest difference between the probabilities that two distributions assign to the same event. We require the joint distribution of $M_T$ and $X$ to be close to the distribution they would have if they were independent.

\begin{definition}[\cite{11339369}]\label{def:approx-individual}
For $T\subseteq[n]$ and $\varepsilon\ge0$, the messages $M_1,\ldots,M_n$ satisfy $\varepsilon$-approximate $T$-individual security with respect to $X$ if $\TV(P_{M_T,X},P_{M_T}P_X)\le\varepsilon$,
where $\TV(P,Q)=\frac12\sum_z|P(z)-Q(z)|$ is the total variation distance. The case $\varepsilon=0$ recovers $T$-individual security.
\end{definition}

For a fixed set $T\subseteq[n]$, let $A_T$ be the matrix formed by the rows indexed by $T$. The following theorem bounds what the entire upload reveals about these rows jointly.

\begin{theorem}\label{thm:rows}
Fix $T\subseteq[n]$ with $|T|\le r$. For Low-Rank Factors, $\TV(P_{A_T,X},P_{A_T}P_X)\le\frac{q^{|T|-r}}{q-1}$. For Low-Rank Ball, $\TV(P_{A_T,X},P_{A_T}P_X)\le\frac{q^{|T|-r}}{q-1}+\beta_{n,r}$, where $\beta_{n,r}$ is the probability that a Low-Rank Ball mask has rank less than $r$. The same bounds hold when $T$ indexes columns.
\end{theorem}

Thus, both sampling methods provide approximate $T$-individual security for any fixed set $T$ of at most $r$ rows or columns. For fixed $q$, the error bound for Low-Rank Factors decreases exponentially with $r-|T|$.

\vspace{-0.2cm}
\subsection{Maximal-Correlation Secrecy}
\label{sec:correlation}
We now ask how much the server can improve its estimates of functions of the inputs. Maximal-correlation secrecy~\cite{8316898} bounds this improvement for all functions, including those that depend on both matrices. Let $M$ denote the data we wish to protect, such as $A$ or the pair $(A,B)$, and let $\mathcal T$ denote the server's observation. We compare a function $f(M)$ of the data with a function $g(\mathcal T)$ computed from the observation.

\begin{definition}[\cite{8316898}]\label{def:correlation}
For real-valued random variables $F,G$ with positive variance, their correlation is
\[
\operatorname{Corr}(F,G)=\frac{\E[FG]-\E[F]\E[G]}{\sqrt{\Var(F)\Var(G)}}.
\]
The maximal correlation $\rh(M;\mathcal T)$ is the supremum of $|\operatorname{Corr}(f(M),g(\mathcal T))|$ over real-valued functions $f,g$ for which $f(M)$ and $g(\mathcal T)$ have positive variance. The protocol provides $\rho$-maximal-correlation secrecy if $\rh(M;\mathcal T)\le\rho$ when $M$ is uniform.
\end{definition}

Maximal correlation is zero exactly when the data and the observation are independent. More generally, a bound $\rh(M;\mathcal T)\le\rho$ means that the observation can reduce the mean-square error in estimating any function of $M$ by at most a $\rho^2$ fraction of the error of the best estimate based only on the prior. We next bound maximal correlation for both sampling methods against the complete server view $\mathcal T=(X,Y,Z)$. Since $Z=XY$ is determined by the uploads, it provides no additional information beyond $X,Y$.

\begin{theorem}\label{thm:rho}
Let $A$ be uniform over $\F^{n\times n}$ and let $1\le r<n$. For either sampling method, with $R$ independent of $A$, we have $\rh(A;A+R)\le q^{-r}$, with equality for Low-Rank Factors. For independent uniform $A,B\in\F^{n\times n}$ and independent masks $R,S$ sampled using the same method and rank parameter $r$, independently of $(A,B)$, we have $\rh((A,B);\mathcal T)=\rh(A;X)$.
\end{theorem}

The bound limits how much the server can improve its guesses about the inputs. For example, over $\mathbb{F}_2$, a given entry of a uniform matrix $A$ is equally likely to be zero or one. Before seeing the uploads, the server can guess its value correctly with probability $1/2$. After seeing the complete view $\mathcal T$, this probability is at most $1/2+2^{-r-1}$. For $r=10<n$, it is less than $50.05\%$, even if the server has unlimited computational power. More generally, for any fixed binary function of $(A,B)$, the improvement over the best guess based on the prior is at most $q^{-r}/2$~\cite{8316898}.

The user can make the server's improvement over the best prior guess tend to zero as the matrices grow, while keeping encoding and decoding close to quadratic. Fix $c>0$, and set $r=\lceil c\log_q n\rceil$. Then, the maximal correlation is at most $n^{-c}$ with $O(n^2\log n)$ field operations for sufficiently large $n$.

We next show that the bound $q^{-r}$ is asymptotically optimal for fixed $q$ and $r=o(n)$. The converse holds for every input-independent additive mask of rank at most $r$, even when combined with secret invertible transformations.

\begin{theorem}\label{thm:converse}
Let $A$ be uniform over $\F^{n\times n}$ and let $X=L_1(A+K)L_2$, where $(L_1,L_2,K)$ is independent of $A$, $L_1,L_2$ are invertible, and $\rank(K)\le r\le n-2$ almost surely. Then
\begin{equation}\label{eq:lower}
\rh(A;X)\ge\frac{q^{2n-r}-2q^n+1}{(q^n-1)^2}\ge\frac{q^{-r}}2.
\end{equation}
\end{theorem}

The server can ignore additional observations, so any view $\mathcal T$ containing $X$ satisfies $\rh((A,B);\mathcal T)\ge\rh(A;X)$. The lower bound therefore also applies to the complete server view. Together with Theorem~\ref{thm:rho}, this shows that both sampling methods are within a factor of two of the optimum for $r\le n-2$. For fixed $q$ and $r=o(n)$, the lower bound is $(1-o(1))q^{-r}$, proving asymptotic optimality within this class.

Although Theorem~\ref{thm:rho} assumes independent uniform inputs, its guessing guarantee extends to nonuniform or correlated inputs. The bound then depends on how concentrated the joint input distribution is. For $M=(A,B)$, define $\Delta=\log_2\bigl(q^{2n^2}\sum_m\Pr[M=m]^2\bigr)$, which is zero for independent uniform inputs. For either sampling method, the improvement over the best prior guess is at most $2^{\Delta/2}q^{-r}$ for any fixed function of $M$, and half this for a binary function~\cite[Thm.~1]{8316898}.

\vspace{-0.2cm}
\section{Proofs}
\label{sec:proofs}
\vspace{-0.2cm}
\begin{proof}[Proof of Theorems~\ref{thm:complexity} and~\ref{thm:posterior}]
Since $Z=AB+AS+RB+RS$, decoding returns $AB$. Computing $RS=U((VQ)W)$ and using $r\le\min\{m,n,p\}$ gives the cost bound. Uniform rank-ball masks give a constant likelihood on compatible inputs and zero elsewhere, so Bayes' rule gives both posterior statements.
\end{proof}

Let $C_{d,e}(t)=\prod_{j=0}^{t-1}\frac{(q^d-q^j)(q^e-q^j)}{q^t-q^j}$ count rank-$t$ matrices and put $V_{d,e}(r)=\sum_{t=0}^rC_{d,e}(t)$. To sample a Low-Rank Ball mask, choose $t$ with probability $C_{d,e}(t)/V_{d,e}(r)$, then multiply independent uniform full-rank factors of sizes $d\times t$ and $t\times e$. Each rank-$t$ matrix has equally many such factorizations, so the product is uniform. Pad the factors to width $r$, i.e., inserting zeros. Rejection sampling takes $O((d+e)r^2)$ expected field operations offline, excluding integer arithmetic for the rank probabilities. Write $C_n=C_{n,n}$ and $V_n=V_{n,n}$.

\begin{proof}[Proof of Proposition~\ref{prop:dp}]
For each entry $i$, let $h_i(K_{-i})$ be a most likely value of $K_i$ given the others and put $p_i=\Pr[K_i=h_i(K_{-i})]$. Apply DP to the event $x_i=h_i(x_{-i})$, comparing zero with every nonzero change to entry $i$. Summing gives $(q-1)p_i\le e^\varepsilon(1-p_i)+(q-1)\delta$. A distribution whose largest probability is $p$ has entropy at least $2(1-p)$ bits. Applying this conditionally and using the chain rule gives $H(K)\ge\sum_i H(K_i\mid K_{-i})
\ge\frac{2n^2(q-1)(1-\delta)}{e^\varepsilon+q-1}$.
Counting factor pairs gives $H(K)\le2nr\log_2q$, proving~\eqref{eq:dp}. For $\delta=0$, singleton events make the support invariant under every entry translation, forcing full support and contradicting $r<n$.
\end{proof}

\begin{proof}[Proof of Theorem~\ref{thm:rows}]
Put $k=|T|$. Since $A$ is uniform, $X$ and $R$ are independent. Given $X=x$, we have $A_T=x_T-R_T$, so the required distance equals the distance of $R_T$ from uniform. Write $p_{k,d}=\prod_{j=0}^{k-1}(1-q^{j-d})$. For Low-Rank Factors, $U_TV$ is uniform when $U_T$ has full row rank, an event of probability $p_{k,r}$. The distance is therefore at most $1-p_{k,r}\le q^{k-r}/(q-1)$.


For uniform rank-$r$ masks, use full-rank factors. Then
$\operatorname{rank}(R_T)=\operatorname{rank}(U_T)$. Conditioning a uniform $n\times r$ matrix $U$ on
having full column rank gives $\Pr[\operatorname{rank}(R_T)=k] = \Pr[\operatorname{rank}(U_T)=k \mid \operatorname{rank}(U)=r] = \frac{p_{k,r}p_{r-k,n-k}}{p_{r,n}} =
\frac{p_{k,r}}{p_{k,n}}$, where the last equality uses
$p_{r,n}=p_{k,n}p_{r-k,n-k}$. Conditional on full row rank, $R_T$ is uniform by invariance under invertible column operations. A uniform $k\times n$ matrix has the same conditional law, with event probability $p_{k,n}$. Both probabilities are at least $p_{k,r}$, giving a distance of at most $1-p_{k,r}$. Mixing lower ranks adds at most $\beta_{n,r}=V_n(r-1)/V_n(r)$. Transposition gives the column bounds.
\end{proof}

\begin{proof}[Proof of Theorem~\ref{thm:rho}]
Let $\chi_H(M)=\psi(\sum_{ij}H_{ij}M_{ij})$, where $\psi$ is a nontrivial additive character of $\F$. These characters form an orthonormal basis, and $\E[\chi_H(A)\mid X=x]=\chi_H(x)\overline{\E[\chi_H(R)]}$. Thus, conditional expectation is diagonal, giving $\rh(A;X)=\max_{H\ne0}|\E[\chi_H(R)]|$. For Low-Rank Factors, averaging over $V$ gives $\E[\chi_H(UV)]=\Pr[H^TU=0]=q^{-r\rank(H)}$, with maximum $q^{-r}$.

For Low-Rank Ball, invariance under invertible row and column operations lets us take $H=\operatorname{diag}(1,H')$ for $H\ne0$. Write $M=\left(\begin{smallmatrix}a&b\\c&D\end{smallmatrix}\right)$. For fixed $b,c,D$ with $\rank(D)<r$, either every $a$ satisfies the rank bound or none does, so the sum over $a$ vanishes. When $\rank(D)=r$, the rank bound requires $b,c$ in its row and column spaces and, after reducing $D$ to $\operatorname{diag}(I_r,0)$, fixes $a=b_1c_1$. Summing $\psi(b_1c_1)$ over $c_1$ vanishes unless $b_1=0$, so the total is $q^r$. Hence $\sum_{\rank(M)\le r}\chi_H(M)=q^r\sum_{\rank(D)=r}\chi_{H'}(D)$.
The absolute value is at most $q^rC_{n-1}(r)$, with equality for rank-one $H$. Each rank-$r$ block $D$ has $q^{2r}$ rank-$r$ extensions, so $V_n(r)\ge q^{2r}C_{n-1}(r)$. Thus $\rh(A;X)=q^rC_{n-1}(r)/V_n(r)\le q^{-r}$.

For independent inputs and masks, the joint coefficients are products with the same largest nonconstant magnitude. Since $Z=XY$, this gives $\rh((A,B);(X,Y,Z))=\rh(A;X)$.
\end{proof}

\begin{proof}[Proof of Theorem~\ref{thm:converse}]
Let $g(A)=|\ker A|=q^{n-\rank(A)}$ and $N=q^n$. Both $A,X$ are uniform, and rank invariance gives $g(X)=g(A+K)$. Counting kernel vectors gives $\E g(A)=2-N^{-1}$ and $\Var(g(A))=(q-1)(1-N^{-1})^2$.

Fix $K$ of rank $t$. We compute $\E[g(A)g(A+K)]$ by counting pairs $(v,w)$ with $Av=0$ and $Aw=-Kw$. There are $(N-1)(N-q)$ linearly independent pairs, each satisfying these equations with probability $N^{-2}$. For a dependent pair $w=cv\ne0$, the equations are consistent exactly when $Kv=0$. This gives $(q-1)(q^{n-t}-1)$ pairs, each contributing $N^{-1}$. The pair $(0,0)$ contributes $1$, and pairs with exactly one zero contribute $2(N-1)/N$. Averaging over $K$, subtracting the squared mean, and dividing by the variance gives $\operatorname{Corr}(g(A),g(X))=\frac{\E[q^{-\rank(K)}]-2N^{-1}+N^{-2}}{(1-N^{-1})^2}$.

Using $\rank(K)\le r$ gives the first bound in~\eqref{eq:lower}; $2q^{-n}\le q^{-r}/2$ for $r\le n-2$ gives the second.
\end{proof}
\vspace{-0.3cm}

\vspace{-0.2cm}
\section{Acknowledgment}
No external funding supported this work. The authors declare no conflicts of interest. 
ChatGPT (OpenAI) was used as an auxiliary tool to identify relevant literature and to check algebraic derivations. All references, mathematical statements, and proofs were independently verified by the authors.

\vspace{-0.2cm}
\section{Compliance with Ethical Standards}
\vspace{-0.2cm}
This work is theoretical and does not involve human or animal subjects. No ethical approval was required.

\vspace{-0.1cm}
\begingroup
\small
\bibliographystyle{IEEEbib}
\bibliography{refs}
\endgroup
\end{document}